\documentclass[12pt,english]{article}
	\usepackage{float}
	\usepackage{setspace}
	\usepackage{amsthm}
	\usepackage{amsmath}
	\usepackage{amssymb}
	\usepackage{graphicx}
	\usepackage[square, numbers]{natbib}
	\usepackage{color}
    \usepackage{parskip}
	\usepackage[unicode=true,pdfusetitle, bookmarks=true,bookmarksnumbered=false,bookmarksopen=false, breaklinks=false,pdfborder={0 0 1},colorlinks=true]{hyperref}
	\hypersetup{citecolor=blue, colorlinks=true}
	\usepackage{lineno}
	
	\newtheorem{proposition}{Proposition}

\newcommand{\beginsupplement}{%
	\setcounter{table}{0}
	\renewcommand{\thetable}{S\arabic{table}}%
	\setcounter{figure}{0}
	\renewcommand{\thefigure}{S\arabic{figure}}%
\renewcommand{\thesection}{S\arabic{section}}   
\renewcommand{\thetable}{S\arabic{table}}   
\renewcommand{\thefigure}{S\arabic{figure}}

}

\begin{document}
 \singlespacing

 \title {Interpreting Economic Complexity }
 \author{
 Penny Mealy\thanks{Institute for New Economic Thinking at the Oxford Martin School, Smith School for Enterprise and the Environment, St Edmund Hall Oxford, OX2 6ED, United Kingdom.}
 \and 
 J. Doyne Farmer\thanks{Institute for New Economic Thinking at the Oxford Martin School, Department of Mathematics, Christ Church College, Oxford, OX2 6ED, United Kingdom.}
 \and
 Alexander Teytelboym\thanks{Department of Economics, St Catherine's College, Institute for New Economic Thinking at the Oxford Martin School, Oxford, OX1 3UQ, United Kingdom.}
 }

 \date{\today}

\maketitle

\begin{abstract}\footnotesize

Two network measures known as the Economic Complexity Index (ECI) and Product Complexity Index (PCI) have provided important insights into patterns of economic development. We show that the ECI and PCI are equivalent to a spectral clustering algorithm that partitions a similarity graph into two parts. The measures are also related to various dimensionality reduction methods and can be interpreted as vectors that determine distances between nodes based on their similarity. Our results shed a new light on the ECI's empirical success in explaining cross-country differences in GDP/capita and economic growth, which is often linked to the diversity of country export baskets. In fact, countries with high (low) ECI tend to specialize in high (low) PCI products. We also find that the ECI and PCI uncover economically informative specialization patterns across US states and UK regions.

\end{abstract}
\newpage

\section{Introduction}
\onehalfspacing

Structural properties of the global trade network can explain differences in economic development across countries \cite{hidalgo2009building, tacchella2012new, caldarelli2012network, hausmann2014atlas,  saracco2015randomizing, straka2017grand}.
A novel pair of measures known as the Economic Complexity Index (ECI) and the Product Complexity Index (PCI) were recently introduced to infer information about countries' productive capabilities from their export baskets \cite{hidalgo2009building, hausmann2014atlas}. These measures have been particularly successful in explaining cross-country differences in GDP/capita and in predicting economic growth. However, the precise mathematical and economical interpretations of these indices have been elusive.

In this paper, we show that the economic complexity measures are mathematically equivalent to a classic spectral clustering algorithm, which partitions a similarity graph into two, balanced components that are internally similar and externally dissimilar \cite{shi2000normalized}. The ECI and PCI can also be interpreted as dimensionality reduction methods, which have close connections to diffusion maps \cite{nadler2006diffusion} and correspondence analysis \cite{mardia1979multivariate,greenacre1984correspondence,yen2011link,beh2014correspondence,hill2014correspondence}

From a dimensionality reduction perspective, the measures define a distance between nodes in a graph based on their similarity. For example, when applied to export data, the ECI (PCI) places countries (products) on a one-dimensional interval such that countries (products) with similar exports (exporters) are close together and countries (products) with dissimilar exports (exporters) are far apart.




Our mathematical interpretations contrast previous conceptual descriptions of the economic complexity measures, which tended to frame the ECI as being related to the diversity (or number) of products a country is able to export competitively \cite{hidalgo2009building, hausmann2014atlas, morrison2017economic, gao2018quantifying}. 
Not only is the ECI mathematically orthogonal to diversity \cite{kemp2014interpretation}, we show that it captures unique, economically insightful information that diversity does not make apparent. When applied to export data, the ECI and PCI reveal a striking pattern of specialization across countries.  High ECI countries (which tend to be richer) specialize in high PCI products, while countries with low PCI (which tend to be poorer) specialize in low PCI products. Moreover, the export baskets of high ECI countries are more homogeneous than the exports baskets of low ECI countries. Hence, while diversity provides information about how many products countries are competitive in, the ECI and PCI shed light on what type of products high- and low-income countries specialize in.

Our results also allow us to extend the ECI and PCI to datasets other than trade data. We provide an illustration with regional data on industrial employment concentrations in UK local authorities and occupational employment concentrations in US states. We find that, remarkably, the ECI for UK local authorities and US states is strongly correlated with regional earnings per capita. Moreover, we show that the ECI and PCI reveal similar patterns of specialization, while diversity fails to be economically informative.

\section{The ECI and PCI}\label{sec:indices}
The ECI and PCI measures are calculated using an algorithm that operates on a binary country-product matrix $M$ with elements $M_{cp}$, indexed by country $c$ and product $p$ \cite{hidalgo2009building}. $M_{cp}=1$ if country $c$ is \textit{competitive} or has a \textit{revealed comparative advantage} (RCA) $> 1$ in product $p$, where RCA is calculated using the Balassa index \cite{balassa1965trade}, given by

\begin{equation}\label{eq:RCA}
RCA_{cp} = \frac{ x_{cp}/ \sum_px_{cp}} { \sum_{c}x_{cp}/ \sum_{c}\sum_p x_{cp} },
\end{equation}
where $x_{cp}$ is country $c$'s exports of product $p$. $M_{cp} = 0$ otherwise. 


Summing across the rows and columns of M gives a country's \textit{diversity} (denoted $k_{c}^{(0)}$) and product \textit{ubiquity} (denoted $k_{p}^{(0)}$), defined as

\begin{equation}
k_{c}^{(0)} = \sum_p M_{cp}
\label{eq:diversity}
\end{equation}
and
\begin{equation}
k_{p}^{(0)} = \sum_c M_{cp}.
\label{eq:ubiquity}
\end{equation}
The ECI and PCI were originally defined through an iterative, self-referential \textit{Method of Reflections} algorithm which first calculates diversity and ubiquity and then recursively uses the information in one to correct the other \cite{hidalgo2009building} (see Methods). However, it can be shown \cite{caldarelli2012network,cristelli2013measuring} that the Method of Reflections is equivalent to finding the eigenvalues of a matrix $\widetilde{M}$, whose rows and columns correspond to countries and whose entries are given by

\begin{equation}\label{eq:Mcc2}
\widetilde{M}_{cc'} \equiv \sum_{p}\frac{M_{cp}M_{c'p}}{k_{c}^{(0)}k_{p}^{(0)}}=\frac{1}{k_{c}^{(0)}}\sum_{p}\frac{M_{cp}M_{c'p}}{k_{p}^{(0)}}.
\end{equation}

Equivalently, we can write $\widetilde{M}$ in matrix notation
\begin{equation}\label{eq:DMUM}
\widetilde{M} = D^{-1}MU^{-1}M',
\end{equation}
where $D$ is the diagonal matrix formed from the vector of country diversity values and $U$ is the diagonal matrix formed from the vector of product ubiquity values. 

When applied to country trade data one can think of $\widetilde{M}$ as a diversity-weighted (or normalized) similarity matrix, reflecting how similar two countries' export baskets are. 

Further, from Eq. \eqref{eq:DMUM}, we can see that 
\begin{equation}\label{eq:DS}
\widetilde{M} =D^{-1}S,
\end{equation}

where $S=MU^{-1}M'$ is a symmetric similarity matrix in which each element $S_{cc'} $ represents the  products that country $c$ has in common with country $c'$, weighted by each the inverse of each product's ubiquity. 

Since $\widetilde{M}$ is a row-stochastic matrix (its rows sum to one) its entries can also be interpreted as conditional transition probabilities in a Markov transition matrix \cite{hidalgo2009building,kemp2014interpretation}. The ECI is defined as the eigenvector associated with the \textit{second-largest} right eigenvalue of $\widetilde{M}$. This eigenvector determines a ``diffusion distance" between the stationary probabilities of states reached by a random walk described by this Markov transition matrix (see the \textbf{SM}).

The PCI is symmetrically defined by transposing the country-product matrix $M$ and finding the second-largest right eigenvalue of $\widehat{M}$, given by
\begin{equation}\label{eq:UMDM}
\widehat{M} = U^{-1}M'D^{-1}M.
\end{equation}

In this paper, we denote the ECI vector by  $\widetilde{y}^{[2]}$ and the ECI of country $c$ is denoted $\widetilde{y}^{[2]}_c$. We also denote the diversity vector by $d$ where $d_c=k_c^{(0)}$ is the diversity of country $c$. Additionally, we note that the ECI is commonly standardized by subtracting the mean and dividing by the standard deviation to allow for comparisons across years \cite{hidalgo2009building,hausmann2014atlas}. However, for clarity, we use the unstandardized ECI vector throughout this paper. 

\section{Results}

The ECI has commonly been described with reference to diversity. This follows from the hypothesis that originally motivated the measure's construction: prosperous countries are likely to be able to competitively a diverse set of products that few other countries are likely to be able to export competitively \cite{hidalgo2009building,hausmann2014atlas}. Recent papers have since described the ECI as an ``indicator of diversity'' \cite[p. 1]{morrison2017economic} and a ``measure of economic diversity'' \cite[p. 1596]{gao2018quantifying}. However, the ECI has been shown to be mathematically orthogonal to diversity \cite{kemp2014interpretation}. That is, the dot product of the diversity and ECI vectors is zero.


The ECI has also been described as a ``standard eigenvalue centrality algorithm'' \citep[p. 1]{morrison2017economic}. However, this description is also inaccurate, as in contrast to the ECI, eigenvector centrality is defined as the eigenvector corresponding to the largest eigenvalue of a symmetric adjacency matrix, such as $S$.\footnote{In the case of directed networks (like $\widetilde{M}$), since the right eigenvector corresponding to the largest eigenvalue is constant,	the natural definition is to take the left eigenvector corresponding to the largest eigenvalue of the adjacency matrix \cite[p. 178]{newman2010networks} (Note that in \cite{newman2010networks}, the exposition of the adjacency matrix is transposed). Moreover, since the rows $\widetilde{M}$ have been normalized by diversity, the leading left eigenvector (eigenvector centrality) will be proportional to diversity and consequently does not add any further information to what we already know about $\widetilde{M}$.}



\subsection{Interpretation as spectral clustering}
We now show that the ECI is mathematically equivalent to a standard spectral clustering method for partitioning an undirected weighted graph, represented by an adjacency matrix $S$, into two balanced components \cite{shi2000normalized}. 

Spectral clustering is a widely used technique for community detection and dimensionality reduction and has a range of diverse applications including image recognition, web page ranking, information retrieval and RNA motif classification. The goal of the spectral clustering approach is to minimize the sum of edge weights cutting across the graph partition, while making the size (number of nodes) of the two components relatively similar. As we discuss below, finding the exact solution to this problem is NP-hard. However, it is possible to obtain an approximate solution by minimizing the \textit{normalized cut (Ncut) criterion} \cite{shi2000normalized}.
We demonstrate that the ECI is equivalent to this approximate solution. 

\subsubsection{The $Ncut$ criterion}
Consider an undirected graph $G = (V, E)$ with vertices $V$ and edges $E$.  We allow the graph $G$ to be weighted, with non-negative weights so the adjacency matrix entries are $S_{ij} \ge 0$ where $S_{ij} = S_{ji}$.  While the export matrix is one possible example, we can consider $S$ to be any matrix with these properties.   The degree of vertex $i$ is defined as
\begin{equation}\label{eq:degrees}
d_i = \displaystyle\sum_{j \in V} S_{ij},
\end{equation}
and the size or ``volume" of a set of vertices $A\subseteq V$ can be measured as
\begin{equation}\label{eq:vol}
vol(A) = \displaystyle\sum_{i\in A}d_{i}.
\end{equation}

(Our notation is deliberate: as we show in the \textbf{SM}, if the adjacency matrix $S$ of the similarity graph $G$ coincides with export similarity matrix $S=D\widetilde{M}$ then degree $d_i$ corresponds precisely to the diversity of a country's exports). One way to partition a graph into two disjoint sets is by solving the {\it cut} problem. The objective is to find a partition of $V$ into complementary sets $A$ and $\bar{A}$ that minimize the number of links between the two sets.  The $cut$ problem is to find the minimum of
\begin{equation}\label{eq:cut1}
cut(A,\bar{A}) = \displaystyle\sum_{i\in A,j\in \bar{A}}S_{ij}.
\end{equation}

This objective function has the undesirable property that its solution often partitions a single node from the rest of the graph. To avoid this problem, the \textit{normalized cut} ($Ncut$) \textit{criterion} \cite{shi2000normalized} penalizes solutions that are not properly balanced. The objective is to partition the graph in such a way that each cluster contains a reasonable number of vertices.  This can be achieved by minimizing the objective function
\begin{equation}\label{eq:cut2}
Ncut(A,\overline{A}) = (\frac{1}{vol(A)} + \frac{1}{vol(\overline{A})} ) \displaystyle\sum_{i\in A,j\in \bar{A}}S_{ij}.
\end{equation}

Let $D$ be the diagonal degree matrix with $D_{ii} = d_i$ and $D_{i \ne j} = 0$. Then finding the minimum value of $Ncut$ is equivalent to solving the optimization problem
\begin{equation}
\label{alternateObjectiveFunction}
\min_ A Ncut(A, \bar{A}) = \min_y\frac{y^T(D-S)y}{y^TDy},
\end{equation}
subject to $y_i \in \left\{1, -vol(A)/vol(\bar{A}) \right\}$ and $y^TD\mathbf{1}=0$.  

Due to the fact that $y_i$ is restricted to one of two possible values, this is not a simple linear algebra problem, and finding the true minimum of the $Ncut$ criterion has been shown to be NP-hard \cite{shi2000normalized}.  However, by letting $y_i$ take on any real value, an approximate solution can be obtained by finding the eigenvector $y^{[2]}$ corresponding to the second-smallest eigenvalue of the generalized eigenvalue equation
\begin{equation}\label{eq:eigvalue}
(D-S)y = \lambda Dy.
\end{equation}
Recall that $L_S=D-S$ is called the {\it Laplacian} matrix of $S$.  By making the substitution 
\begin{equation}
\label{substitution}
y = D^{-1/2} z,
\end{equation}
this can be rewritten as a standard eigenvalue equation
\begin{equation}
\label{normalizedLaplacian}
D^{-\frac{1}{2}}(D-S)D^{-\frac{1}{2}}z = \overline{L}_S z = \lambda z,
\end{equation}
where $\overline{L}_S=D^{-\frac{1}{2}}(D-S)D^{-\frac{1}{2}}$ is the \textit{normalized Laplacian} of $S$. Because the normalized Laplacian is a stochastic matrix, its smallest eigenvalue is zero. The eigenvector $z^{[2]}$ associated with the second-smallest eigenvalue of $\overline{L}_S$ is called the {\it normalized Fiedler vector}, and is a solution to the standard eigenvalue equation in Eq.~(\ref{normalizedLaplacian}).
%
%
Transforming back to $y$ using Eq.~(\ref{substitution}) to solve the original problem gives the solution
\begin{equation}
\label{NcutSolution}
y^{[2]} = D^{-1/2} z^{[2]}.
\end{equation}
The solution $y^{[2]}$ provides a useful approximate solution that minimizes the normalized cut criterion and is equal to a simple transformation of the normalized Fiedler vector \cite{shi2000normalized}.

\subsubsection{The relationship between the ECI and the $Ncut$ criterion}
Recall that $\widetilde{M}$ is the matrix whose eigenvector corresponding to the second-largest eigenvalue is the ECI. To see the relationship between spectral clustering and the ECI, note that the similarity matrix $S = D \widetilde{M}$ characterising country export similarity is in the same form used to minimize the normalized cut criterion. 
Multiplying both sides of Eq.~(\ref{normalizedLaplacian}) by $D^{-\frac{1}{2}}$ and re-arranging terms gives
\begin{equation}
D^{-1} S D^{-\frac{1}{2}} z = (1 - \lambda) D^{-\frac{1}{2}} z.
\end{equation}
Substituting $\widetilde{M} = D^{-1}S$ gives
\begin{equation}\label{eq:eigentrans}
\widetilde{M} D^{-\frac{1}{2}} z = (1 - \lambda) D^{-\frac{1}{2}} z.
\end{equation}
The eigenvalue equation for $\widetilde{M}$ is
\begin{equation}\label{eq:eigen}
\widetilde{M} \widetilde{y} = \widetilde{\lambda} \widetilde{y}.
\end{equation}
Now, comparing Eqs.~\eqref{eq:eigentrans} and \eqref{eq:eigen}, we can see that the eigenvalues and eigenvectors of $\widetilde{M}$ are related to those of $\overline{L}_S$ by
\begin{eqnarray}
\widetilde{\lambda}&=& 1 - \lambda,\text{ and}\\
\widetilde{y}&=& D^{-\frac{1}{2}} z.
\end{eqnarray}
Thus the second-smallest eigenvalue of  $\overline{L}_S$  corresponds to the second-largest eigenvalue of $\widetilde{M}$, and comparison to Eq.~(\ref{NcutSolution}) makes it clear that the \textit{ECI is equivalent to the spectral clustering solution of the normalized cut criterion}, i.e.

\begin{equation}\label{ECIisNcut}
\widetilde{y}^{[2]} = y^{[2]}=D^{-\frac{1}{2}} z^{[2]},
\end{equation}
where $\widetilde{y}^{[2]}$ represents the second largest eigenvector of $\widetilde{M}$.

That is, the ECI ($\widetilde{y}^{[2]}$) is equivalent to the approximate solution ($y^{[2]}$) that minimizes the normalized cut criterion. Moreover, the ECI is related to the normalized Fiedler vector by a simple transformation. In the \textbf{SM}, we also show how this interpretation can be applied to the PCI, and describe the mathematical relationship between the ECI and PCI. 

\subsubsection{Applying the spectral clustering interpretation to economic data}

We now demonstrate how the  ECI partitions similarity networks in practice. A visual illustration is shown in Panel A of Figure \ref{fig:Sgraphs}. Here we have calculated the ECI for a randomly generated similarity graph with two clear components. The ECI assigns each node a real number on an interval with positive and negative values according to their similarity to each other. On the left plot of Panel A, we show the ECI values associated with each node in ascending order. The graph should be partitioned where ECI is zero.  Nodes with a positive ECI are assigned to one cluster and nodes with a negative ECI go into the other cluster. In this case, the distinct gap in the ECI values shows that the partition is very clear. On the right plot of Panel A, we show the network's adjacency matrix $S$, where we have also ordered the rows and columns in accordance with the ascending ECI values. Here one can also see how the ECI ordering reveals the graph's two clear components. 

\begin{figure}[H]
	\centering
	\includegraphics[width=1\textwidth]{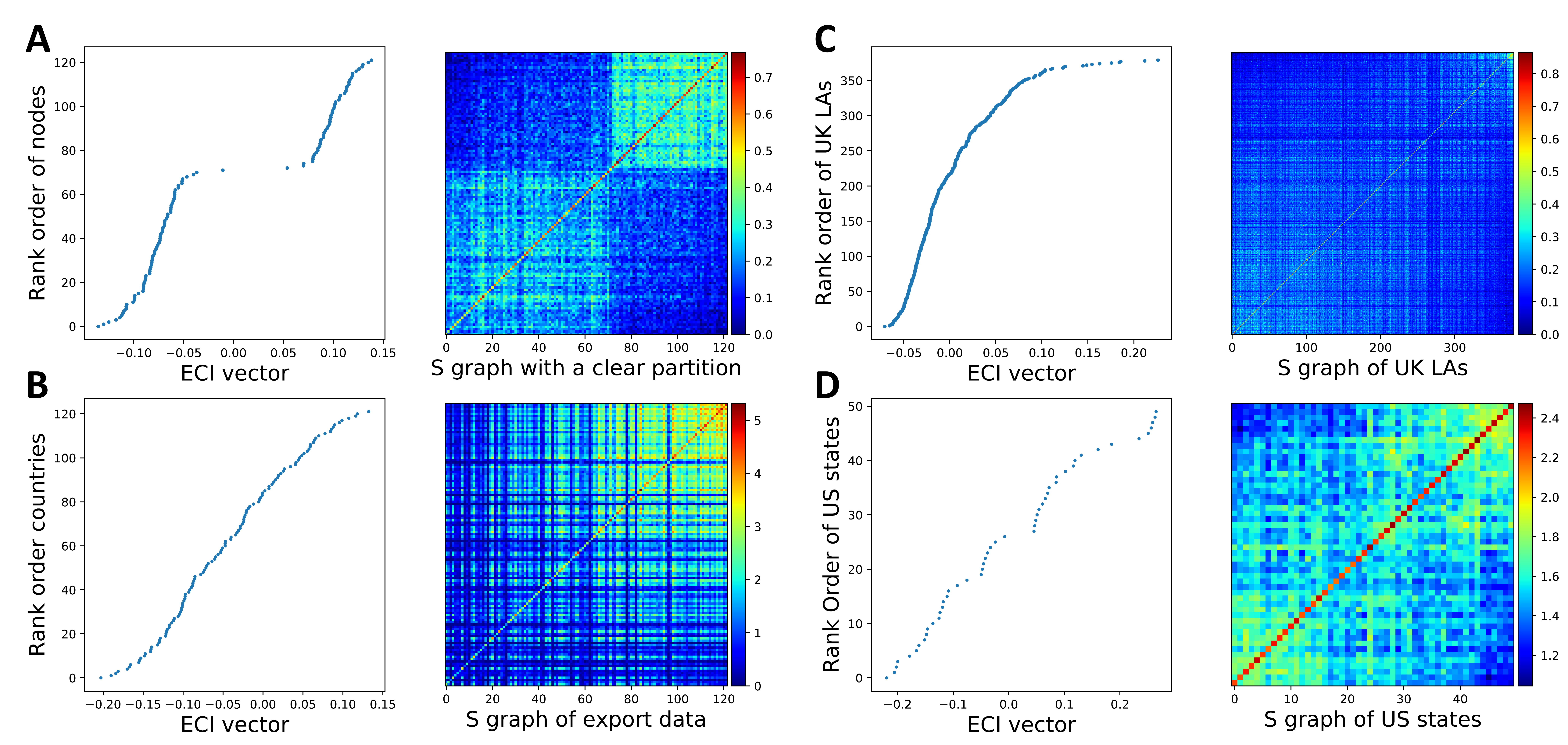}
	\caption{Each panel shows the ECI vector (in ascending order) on the left plot and the associated similarity matrix $S$ on the right plot, where rows and columns have been ordered by the ECI and colored by the $S_{ij}$ values. Panels correspond to similarity networks based on:
		\textbf{Panel A)} randomly generated data with two clear components; \textbf{Panel B)} HS6 COMTRADE data for 2013; \textbf{Panel C)} data on employment concentrations in different industries in  UK local authorities; \textbf{Panel D)} data on employment concentrations in different occupations in US states.}
	\label{fig:Sgraphs}
\end{figure}

In Panel B, we show the same for export data (based on HS6 COMTRADE data for the year 2013). In the left plot, country ECI values (sorted in ascending order) do not show a clear gap across the zero threshold. Moreover, the plot on the right suggests that while countries with high ECI values have a high degree of similarity in their exports (as shown by the higher $S_{ij}$ values), countries with low ECI values appear to have more varied export portfolios. These plots therefore indicate that the export data do not partition cleanly into two components. 

In Panel C and D, we apply the ECI to two other similarity networks constructed from regional data for the UK and US.  Panel C shows a similarity graph constructed on the basis of regional data from the UK Business Register and Employment Survey (BRES) for the year 2011 (available from \url{https://www.nomisweb.co.uk/}). Here, nodes are UK local authorities, which are similar to each other on the basis of their employment concentrations in different industries (classified at the three-digit level of granularity). The similarity graph in Panel D is constructed from regional data sourced from the Integrated Public Use Microdata Series (IPUMS) \cite{ruggles2017ipums} for the year 2010 (available from \url{https://usa.ipums.org/usa/}). In this graph, nodes are US states, and similarity is calculated on the basis of employment concentrations in different occupations (also classified at the three-digit level of granularity). More detail about the construction of these networks can be found in the Materials and Methods section.

Interestingly, in both of these examples, the data also do not partition clearly into two components. However, as we show in the next section, the ECI and PCI  nonetheless provides economically insightful information.

\subsection{Interpretation as a dimensionality reduction tool}\label{sec:dimred}
In addition to approximating the normalized cut criterion, the ECI can be interpreted as a dimensionality reduction tool. As Shi and Malik \cite{shi2000normalized} show, the ECI exactly minimizes 
\begin{equation}
\label{alternative2}
\frac{\sum_{ij} \left( y_i - y_j \right)^2 S_{ij}} {\sum_i y_i^2 d_i},
\end{equation}
subject to the constraint
\begin{equation}
\label{orthog}
\sum_i y_i d_i = 0.
\end{equation}

Here, the objective is to find real numbers $y_i$ for each node $i$ that minimize the sum of the squared distances between nodes, where the distances are weighted according to the similarity matrix $S$. The constraint ensures that the assigned $y_i$ numbers take on positive and negative values and are reasonably balanced in their distribution above and below zero. As we will discuss further in section \ref{sec:revisit}, it also hard-wires the orthogonality condition between the ECI and diversity vectors. 

When applied to export data, we can interpret the ECI as a method to collapse the high-dimensional space of country-export similarities into one dimension. The ECI positions countries on an interval where similar countries are placed close together and dissimilar countries are placed far apart. The distance between countries on this line is a special case of the ``diffusion map distance'' \cite{nadler2006diffusion} and is closely related to correspondence analysis (see \cite{yen2011link} and the \textbf{SM}).

The application of economic complexity measures to export data is particularly interesting from an economic perspective because the ECI strongly correlates  with countries' per capita GDP and future growth rates \cite{hidalgo2009building,hausmann2014atlas}. However, as we show in Figure \ref{fig:eci_gdp}, the ECI yields economically insightful information beyond export data (Panel A). Panel B shows that the ECI for UK local authorities is also correlated with per capita earnings and Panel C shows that the ECI for US states correlates with state-level per capita GDP\footnote{UK earnings data is sourced from the UK Office for National Statistics Annual Survey of Hours and Earning and US state-level per capita GDP data is sourced from the US Bureau of Economic Analysis.}.

\begin{figure}[H]
	\centering
	\includegraphics[width=1\textwidth]{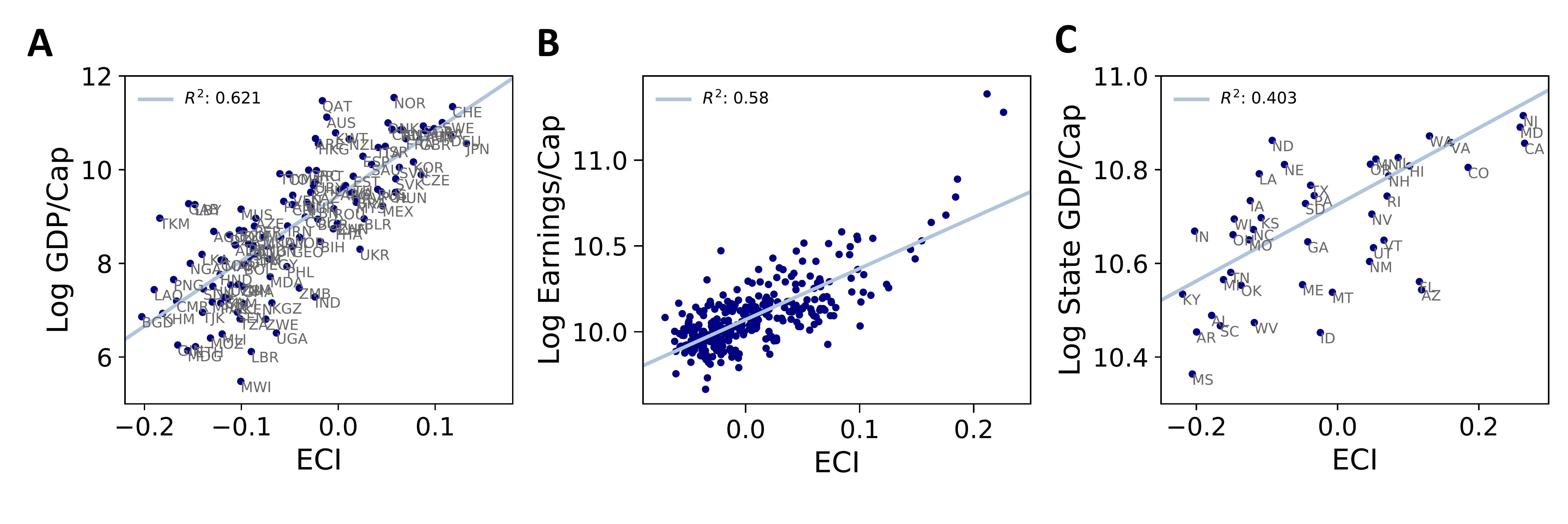}
	\caption{\textbf{Panel A)} Relationship between the ECI and log GDP per capita for data on countries and exports. \textbf{Panel B)} Relationship between the ECI and log per capita earnings for data on industrial employment concentrations in UK local authorities. As the scatter plot is too tightly clustered to show legible local authority labels, we provide the top and bottom 10 local authorities ranked by their ECI in the \textbf{SM}. \textbf{Panel C)} Relationship between the ECI and log GDP per capita for for data on occupational employment concentrations in US States. }
	\label{fig:eci_gdp}
\end{figure}

Applying the PCI to the country export space provides additional economic insights. Analogous to the ECI, the PCI is defined as the  eigenvector associated with the second largest eigenvalue of the transpose of the $\widetilde{M}$ matrix. Hence, the PCI places products along a one-dimensional interval such that products exported by the same countries are close together and products exported by different countries are far apart. Moreover, as the ECI of a country is equal to the average of the PCI of products the country is competitive in (see the \textbf{SM}), the PCI sheds light on the type of exports that countries have in common.

 We provide an illustration in Figure \ref{fig:banding}, which shows the $M$ (bipartite) matrix for countries and exports (Panel A), UK local authorities and industries (Panel B) and US states and occupations (Panel C). In all three cases, we order the country or region rows in accordance with their corresponding ECI (sorted in ascending order). We also order the export, industry and occupation columns by their corresponding PCI (also sorted in ascending order). 

\begin{figure}[H]
	\centering
	\includegraphics[width=1\textwidth]{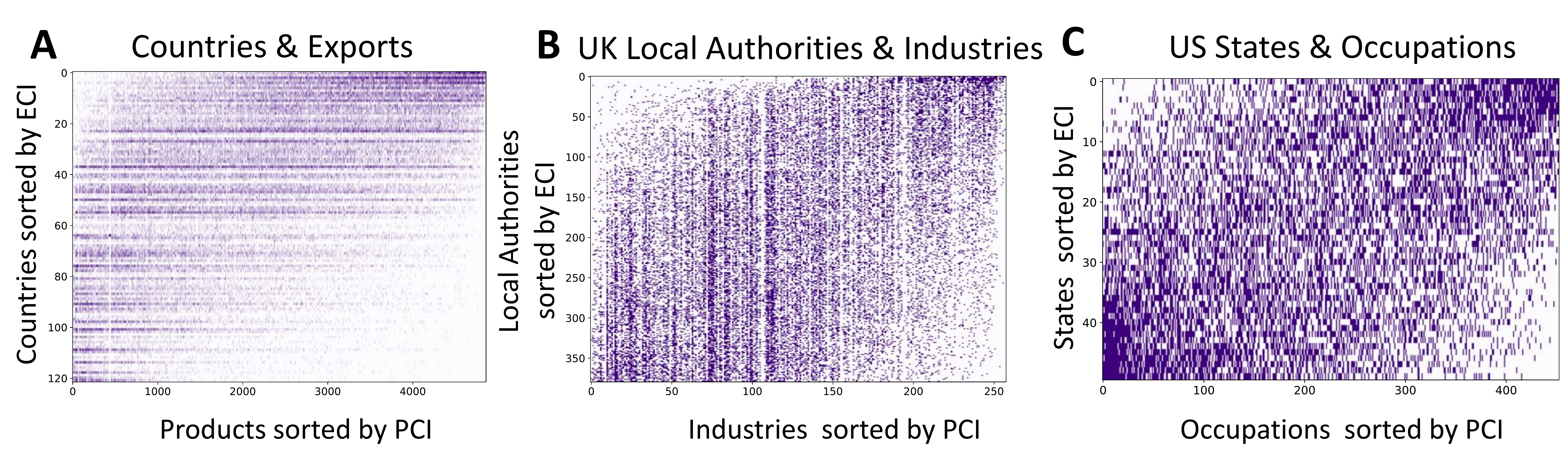}
	\caption{In each matrix, rows are sorted by the ECI and columns are sorted by the PCI: \textbf{Panel A)} country-product $M$ matrix; \textbf{Panel B)} UK region-industry $M$ matrix; \textbf{Panel C)} US state-occupation $M$ matrix.}
	\label{fig:banding}
\end{figure}

Panel A reveals a striking pattern of specialization in the export data. By simultaneously looking at ECI and PCI we can infer that rich and poor countries differ systematically according to the types of exports they are competitive in. Richer (poorer) countries with high (low) ECI specialize in high (low) PCI products. Remarkably, similar specialization patterns are also evident in the UK and US regional data. In the \textbf{SM}, we show the top and bottom local authorities and US states ranked by ECI, as well as the top and bottom industries and occupations ranked by PCI. In the UK, high (low) ECI local authorities tend to be urban (rural) areas specialised in high (low) PCI industries relating to financial and professional (agricultural and manufacturing) industries. We find similar results for the US. Hence, by interpreting the ECI and PCI as similarity measures, we are able to uncover new economic insights from well-known and new datasets. 

\subsection{Revisiting previous interpretations of economic complexity}\label{sec:revisit}

Previous interpretations of the ECI have tended to be cast in terms of diversity \cite{hidalgo2009building,hausmann2014atlas, morrison2017economic,gao2018quantifying}, even though the ECI and diversity are mathematically orthogonal (see Eq.~\eqref{orthog} and \cite{kemp2014interpretation}). However, in the country-export data (see Panel of A of Figure \ref{fig:div_eci}) as well as in Chinese regional data \cite{gao2018quantifying}, diversity and the ECI turn out to be positively correlated. Recall that orthogonality (having a zero dot product) does not imply zero correlation unless the mean of one of the variables is zero. Neither diversity nor the (unstandardized) ECI have zero means in these data. Indeed, as we show in Panel B and C of Figure \ref{fig:div_eci}, the empirical relationship between the ECI and diversity is different in the UK and US regional data. Despite being positively correlated with regional per capita earnings (Figure \ref{fig:eci_gdp}), the ECI is negatively correlated with  industrial diversity of UK local authorities and has no significant correlation with occupational diversity of US states.

\begin{figure}[H]
	\centering
	\includegraphics[width=1\textwidth]{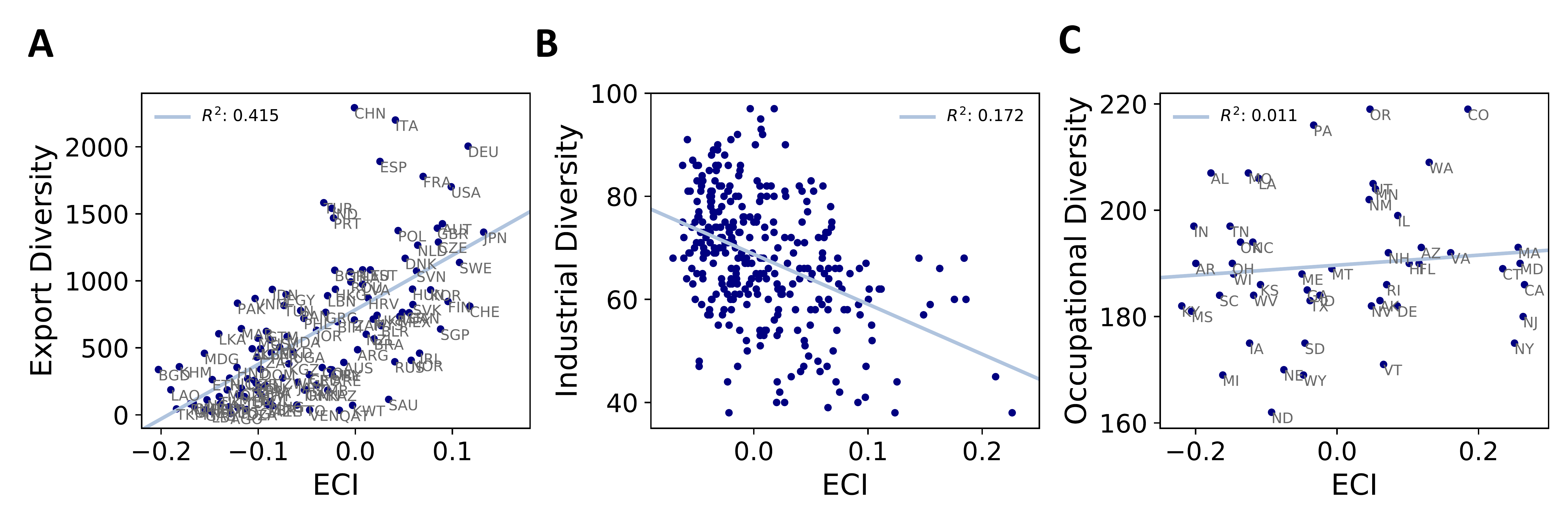}
	\caption{Relationship between diversity and the ECI for data on: \textbf{Panel A)} countries and exports; \textbf{Panel B)} UK regions and industries; \textbf{Panel C)} US states and occupations. }
	\label{fig:div_eci}
\end{figure}

The mathematical orthogonality between the ECI and diversity indicates that these variables capture different information \cite{kemp2014interpretation}. In the export data, the ECI and diversity both provide useful economic insights. In particular, previous work has shown that ordering the rows of matrix $M$ by country diversity and the columns by product ubiquity reveals a triangular structure \cite{hausmann2011network} (see Panel A of Figure \ref{fig:triangular}). This pattern indicates that more diverse countries tend to export less ubiquitous products while less diverse countries tend to export more ubiquitous products in a sharp contrast to traditional theories of comparative advantage \cite{hausmann2011network}.

\begin{figure}[H]
	\centering
	\includegraphics[width=1\textwidth]{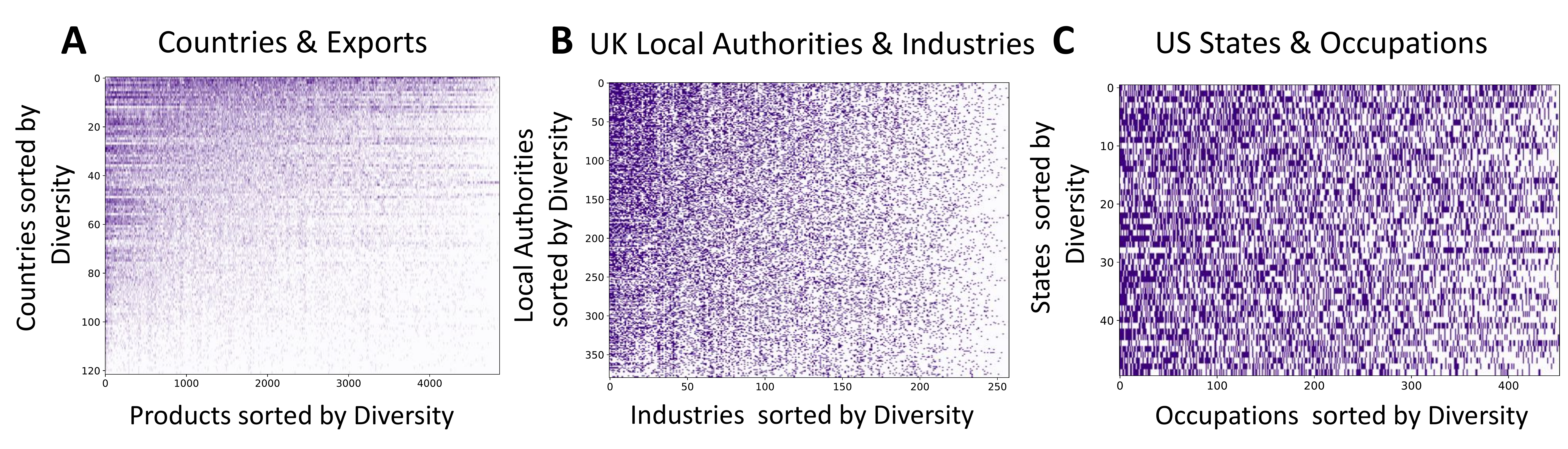}
	\caption{In each matrix, rows are sorted by diversity and columns are sorted by ubiquity: \textbf{Panel A)} country-product $M$ matrix; \textbf{Panel B)} UK region-industry $M$ matrix; \textbf{Panel C)} US state-occupation $M$ matrix.}
	\label{fig:triangular}
\end{figure}

However, in both of our regional examples, diversity and ubiquity fail to be economically informative. As we can see in Panels B and C of Figure \ref{fig:triangular}, the diversity and ubiquity ordering of matrix $M$ constructed from US and UK regional data does not reveal a triangular structure. Moreover, as shown in Figure \ref{fig:div_gdp}, while country diversity is positively correlated with per capita GDP in the export data (Panel A), there is no positive correlation between diversity and per capita earnings in the UK (Panel B) or per capita state-level GDP in the US (Panel C).

\begin{figure}[H]
	\centering
	\includegraphics[width=1\textwidth]{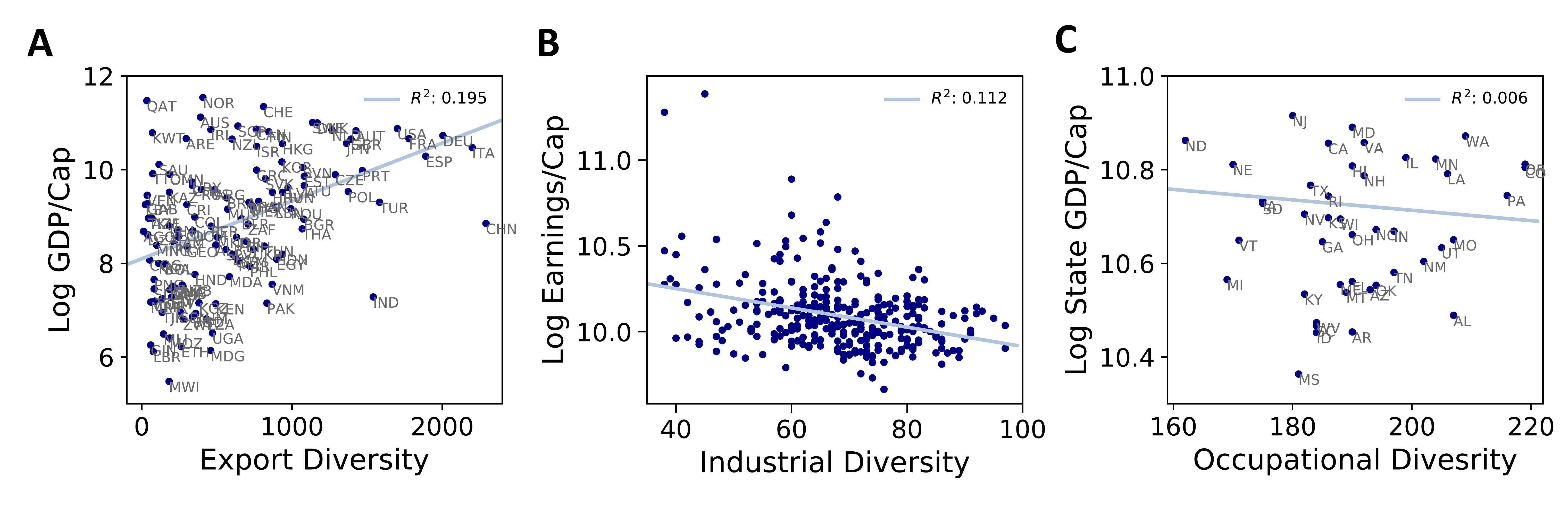}
	\caption{\textbf{Panel A)} Relationship between diversity and log GDP per capita for data on countries and exports. \textbf{Panel B)} Relationship between diversity and log per capita earnings for data on industrial employment concentrations in UK local authorities. \textbf{Panel C)} Relationship between diversity and log GDP per capita for for data on occupational employment concentrations in US states. }
	\label{fig:div_gdp}
\end{figure}

\section{Discussion}

This paper provides a number of mathematical interpretations of the ECI and PCI and shows how these interpretations offer useful economic insights in export and regional data. Our results also cast existing empirical findings in a new light. Previously, the success of the ECI in explaining variation in per capita GDP and future growth rates across countries was thought to reflect the importance of accumulating a diverse set of productive capabilities \cite{hidalgo2009building,hausmann2011network,hausmann2014atlas}. However, by making the difference between the ECI and diversity explicit, we can better understand the distinct roles these variables play in the development process. 

The relationship between diversification and development is well established in the economics literature. Countries tend to follow a U-shaped pattern, whereby they first diversify and then begin specializing relatively late in the development process \cite{imbs2003stages}. This pattern aligns with other empirical studies that have described a positive association between export diversification and economic growth, which tends to be stronger for less developed countries \cite{al2000export, herzer2006does, hesse2009export}.

In contrast to diversity, the ECI and PCI reveal additional information about the type of exports that countries at lower and higher income levels specialize in. High PCI products (which tend to be exported by richer, high ECI countries) have been shown to relate to chemical and machinery exports that require technologically sophisticated know-how and advanced manufacturing processes, while low PCI products (which tend to be exported by poorer, low ECI countries) correspond to simple agricultural products or raw minerals  \cite{hausmann2014atlas}. While the importance of technological upgrading for growth and development is also well recognized within the economics literature \cite{lall1992technological,lall2000technological,lall2006sophistication}, our interpretation of the ECI and PCI as similarity measures sheds new empirical light on the technological differences in the export-baskets of low- and high-income countries. 

The mathematical connections between the economic complexity measures and spectral clustering also open the door for further applications of dimensionality reduction methods to other economic datasets. Indeed, as we have shown with our illustration of UK and US employment data, the ECI and PCI reveal similar patterns of specialization across richer and poorer regions. Interestingly, in these two particular examples, we find that the type -- rather than the diversity -- of industries and occupations concentrated in a region appears to matter more for regional economic prosperity. Future work could readily extend the economic complexity measures to examine other economic networks, such as production networks constructed from country input-output data.  Moreover, the relationships between the ECI, diffusion maps, \cite{coifman2006diffusion,yen2011link}, and simple correspondence analysis \cite{zha2001bipartite} (some of which are discussed in the \textbf{SM}), suggest that new insights could be gleaned from applications of nonlinear diffusion maps and multiple correspondence analysis to economic data.

\section{Materials and Methods}
\subsection{Calculating the ECI for UK and US regional employment data}

\subsubsection{UK Local Authorities and Industries}
Using data from the UK Business Register and Employment Survey (BRES), we construct a binary \textit{region-industry matrix} $W$ on the basis of a region $r$'s \textit{Location Quotient} (LQ) in industry $i$

\begin{equation}\label{eq:LQ}
LQ_{ri} = \frac{ e_{ri}/ \sum_ie_{ri}} { \sum_{r}e_{ri}/ \sum_{r}\sum_i e_{ri} },
\end{equation}
where $e_{ri}$ is the number of people employed in industry $i$ in region $r$ and $W_{ri} = 1$ if $LQ_{ri} > 1$ and $LQ_{ri} = 0$ otherwise. Note that Eq.~\eqref{eq:LQ} is analogous to Eq.~\eqref{eq:RCA}. We then construct a $\widetilde{W}$ matrix from $W$ in the same way as $\widetilde{M}$ is constructed from $M$ (Eq.~\ref{eq:DMUM}). Finally, we calculate the \textit{industry-based} ECI for UK Local Authorities by finding the eigenvector associated with the second-largest eigenvalue of $\widetilde{W}$.

\subsubsection{US States and Occupations}
We apply the same methodology to calculate the \textit{occupation-based} ECI for US states. (We also find consistent results using data on US states and industries.) Drawing on census data for the US, which is available from the Integrated Public Use Microdata Series (IPUMS) \cite{ruggles2017ipums}, we construct a \textit{state-occupation matrix} using state's location quotient in \textit{occupation} $i$. We then compute the \textit{occupation-based} ECI for US states analogously to the industry-based ECI for UK Local Authorities.
\newpage
\beginsupplement

\section*{Supplementary Material}
\setcounter{section}{0}
\section{Diversity and degree equivalence}
Recall that diversity is defined as:
\begin{equation}
	k_{c}^{(0)} = \sum_p M_{cp}
	\label{eq:diversity}
	\end{equation}
and the degree of a node in a graph defined by a similarity matrix $S$ is
\begin{equation}\label{eq:degrees}
	d_i = \displaystyle\sum_{j} S_{ij},
	\end{equation}
We want to show that these are equivalent. We defined $S=MU^{-1}M'$. Note that $U^{-1}M'$ is row-stochastic and $D^{-1}M$ is also row-stochastic. Therefore, any row of $\widetilde{M}=D^{-1}MU^{-1}M'$ adds up to 1 and hence every row $i$ of $MU^{-1}M'$ must add up to $D_{ii}$.

\section{Relationship between the ECI and PCI}
\begin{proposition}
	A country's ECI is equal to the average PCI of products that the country has revealed comparative advantage in.
\end{proposition}
\begin{proof}
	
	Recall that 
	\begin{equation}\label{eq:mtilde}
	\widetilde{M} = D^{-1}MU^{-1}M'.
	\end{equation}
	
	The ECI is one of the solutions $\widetilde{y}$ to the following eigensystem:
	\begin{equation}\label{eq:eigenECI}
	\widetilde{M}\widetilde{y} = \widetilde{\lambda}\widetilde{y}.
	\end{equation}
	
	To calculate the PCI for all products, we are interested in the second eigenvector of the matrix $\widehat{M}$, which is given by
	
	\begin{equation}
	\widehat{M} = U^{-1}M'D^{-1}M.
	\end{equation}
	
	Hence, PCI is one of the solutions $\widehat{y}$ to the following eigensystem:
	\begin{equation}\label{eq:PCI}
	\widehat{M}\widehat{y} = \widehat{\lambda}\widehat{y}.
	\end{equation}
	To prove the proposition, take Eq.~\eqref{eq:eigenECI}, the eigensystem for ECI, and substitute in Eq. \eqref{eq:mtilde}:
	\begin{equation}
	D^{-1}MU^{-1}M'\widetilde{y} = \widetilde{\lambda}\widetilde{y},
	\end{equation}
	\begin{equation}
	M^{-1}DD^{-1}MU^{-1}M'\widetilde{y} = \widetilde{\lambda}M^{-1}D\widetilde{y},
	\end{equation}
	\begin{equation}
	U^{-1}M'\widetilde{y} = \widetilde{\lambda}M^{-1}D\widetilde{y},
	\end{equation}
	which is equivalent to the eigensystem for PCI, 
	\begin{equation}
	U^{-1}M'D^{-1}M\widehat{y} = \widehat{\lambda}\widehat{y}
	\end{equation}
	for 
	\begin{equation}
	\widetilde{y}=D^{-1} M \widehat{y}.
	\end{equation}
	as required.
\end{proof}

Therefore, the ECI can be immediately obtained from the PCI by using $M$. Moreover, note that all the eigenvalues of $\widehat{M}$ and $\widetilde{M}$ are the same.

\section{Interpretation of ECI as a diffusion map and relationships to correspondence analysis and kernel principal component analysis}
A \emph{diffusion map} is a dimensionality reduction method that generates representations of complex data sets in a lower-dimensional Euclidean space by iterating the Markov matrix associated with the data \cite{coifman2006diffusion,nadler2006diffusion}. Since $\widetilde{M}$ can be seen as a Markov transition matrix (see section \ref{sec:indices}), the ECI can also be used to construct a \emph{basic diffusion map} that indicates how a random walker beginning at a particular node (or Markov chain ``state'') will move through the system \cite{yen2011link}. 

For example, if we let the nodes in graph $S$ represent states in a Markov transition matrix, the probability that a random walk beginning in state $i$ reaches state $j$ in the next step is given by $\widetilde{M}_{ij}$. Now consider two random walks beginning in states $i$ and $j$. How ``far'' the random walks are from each other at time $t$ tells us something about the similarity of nodes $i$ and $j$ in graph $S$. Let vector ${x}_i(t)$ denote the probability distribution over states reached at time $t$ by a random walk beginning in state $i$. Then define the 
\emph{diffusion map distance} to be proportional to 

\begin{equation}
({x}_i(t)-{x}_j(t))' D^{-1}(({x}_i(t)-{x}_j(t)).
\end{equation}

Each states at time $t$ can be represented as a point in an $n$-dimensional Euclidean space with coordinates \[(|\lambda^t_2|y^{[2]}_i,|\lambda^t_3|y^{[3]}_i,\ldots, \lambda^t_n|y^{[n]}_i)\]
where $\lambda_j$ is the eigenvalue associated with the $j^{\text{th}}$ largest eigenvector $\widetilde{M}$ and $y^{[n]}_i$ is the $i^{\text{th}}$ entry of the $n^{\text{th}}$ largest eigenvector of $\widetilde{M}$ \cite{yen2011link}. The distance between the points is precisely the diffusion map distance.

In Figure \ref{fig:diff_map}, we apply the diffusion map to country export data. By using the second and third coordinates of the diffusion map, we visualize countries in a two-dimensional plane at different $t$. Since the second largest eigenvalue is dominant, we rescale the axis by its value. As $t$ goes to infinity, the diffusion map distance captures the distance between the stationary probabilities of states in the random walk and is well approximated by the second-largest eigenvector of $\widetilde{M}$ i.e. ECI \cite{nadler2006diffusion}.

\begin{figure}[H]
	\centering
	\includegraphics[width=1\textwidth]{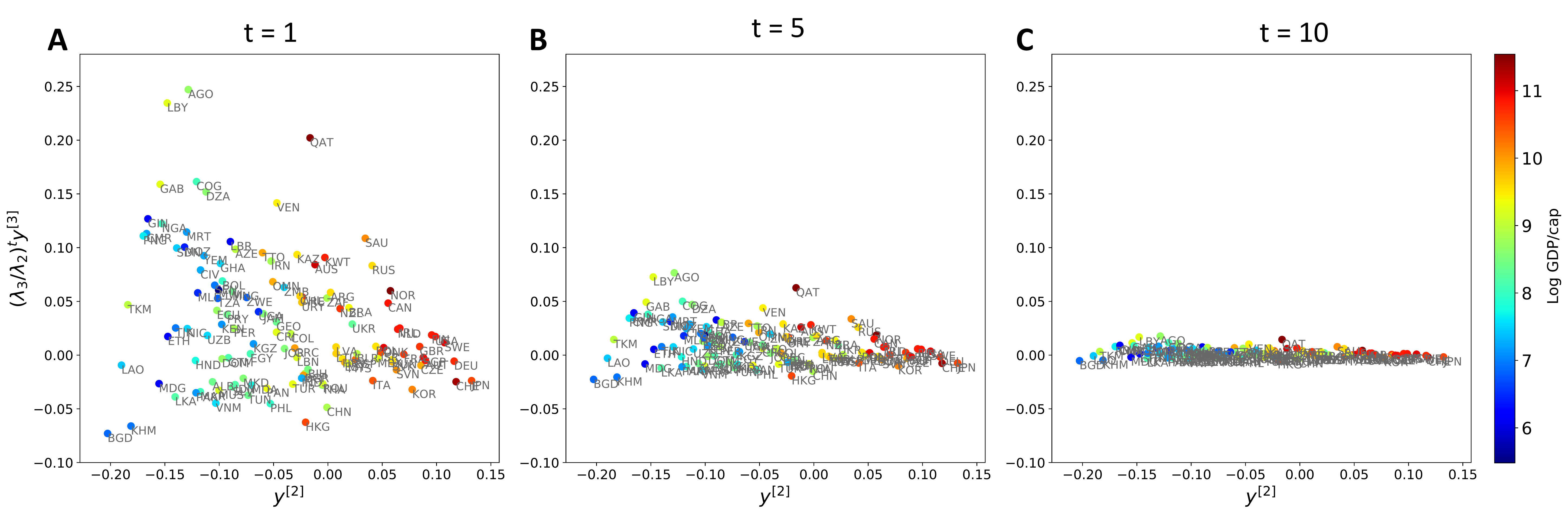}
	\caption{Application of diffusion map interpretation to country export data.}
	\label{fig:diff_map}
\end{figure}

There is also an equivalence between the $Ncut$ criterion and \emph{correspondence analysis} (CA) \cite{zha2001bipartite}. Simple (multiple) CA is a classic tool in multivariate analysis that studies relationships between two (two or more) categorical variables, such as countries and products, via singular value decomposition \cite{mardia1979multivariate,greenacre1984correspondence,beh2014correspondence,hill2014correspondence}. In this setting, the similarity matrix $S$ represents the \emph{Pearson correlation matrix}. Performing simple correspondence analysis is equivalent to computing the basic diffusion map when $t=1$ \cite{yen2011link}.

Finally, diffusion maps are also related to \emph{kernel Principal Component Analysis} (PCA). Define 
\begin{equation}
	K(t)=\widetilde{M}^tD^{-1}\widetilde{M}^{'t}
	\end{equation}
which is a symmetric, positive-definite matrix known as the \emph{diffusion map kernel}. Denote $w^{[n]}$ to be an eigenvector of $K$ associated with $\mu_n$, the $n^\text{th}$ largest eigenvalue. Each state at time $t$ can be represented in an $n$-dimensional Euclidean space with coordinates 
\begin{equation} (\sqrt{\mu_1}w^{[1]}_{i},\sqrt{\mu_2}w^{[2]}_{i},\ldots,\sqrt{\mu_n}w^{[n]}_{i}).
\end{equation}  

This is not only a vector representation of each in the principal component space, but also the distance between the points is exactly the diffusion map distance. 

A clear summary of relationships between different spectral methods for computing a low-dimensional embedding of undirected weighted graphs can be found in \cite[Table 10.1, p. 439]{fouss2016algorithms}.


\section{ECI and PCI Rankings for Regional Data}
In this section, we show the top and bottom ECI and PCI rankings for UK local authorities (Table \ref{tab:eci_LA_table}), UK industries (Table \ref{tab:pci_industry_table}), US states (Table \ref{tab:eci_state_table}) and US occupations (Table \ref{tab:pci_occ_table}). 


\begin{table}[H]
	\centering
	\includegraphics[width=0.7\textwidth]{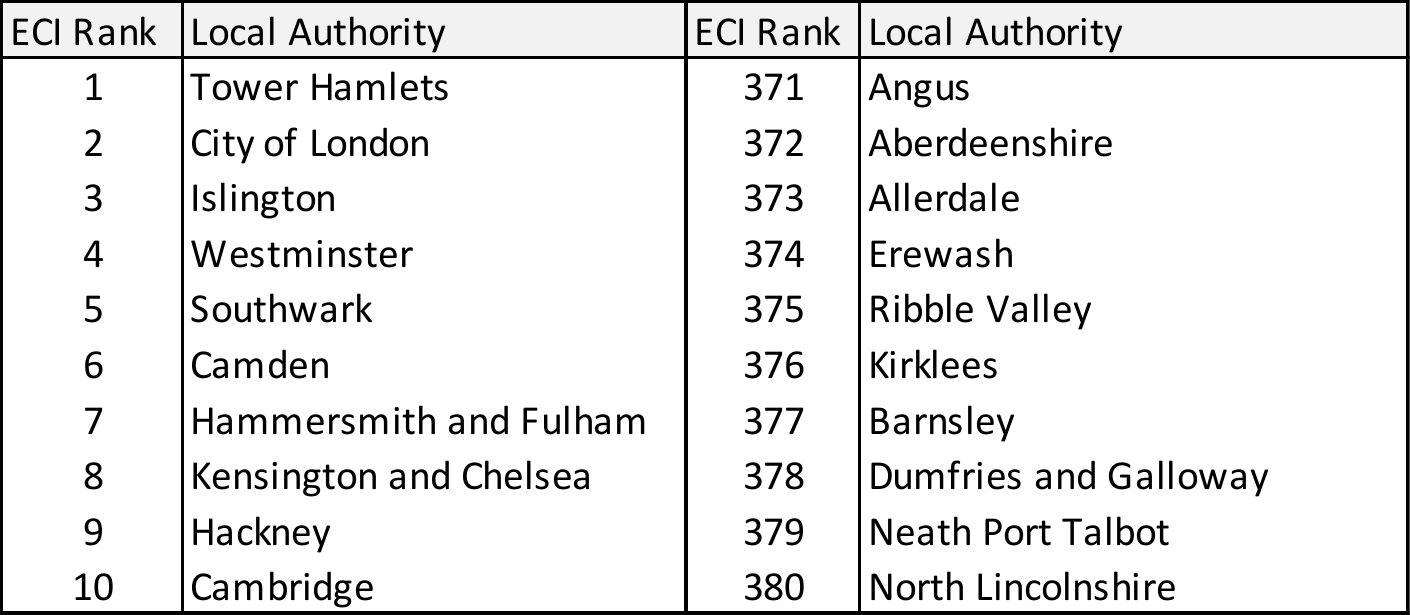}
	\caption{Top and bottom 10 UK local authorities ranked by ECI.}
	\label{tab:eci_LA_table}
\end{table}

\begin{table}[H]
	\centering
	\includegraphics[width=1\textwidth]{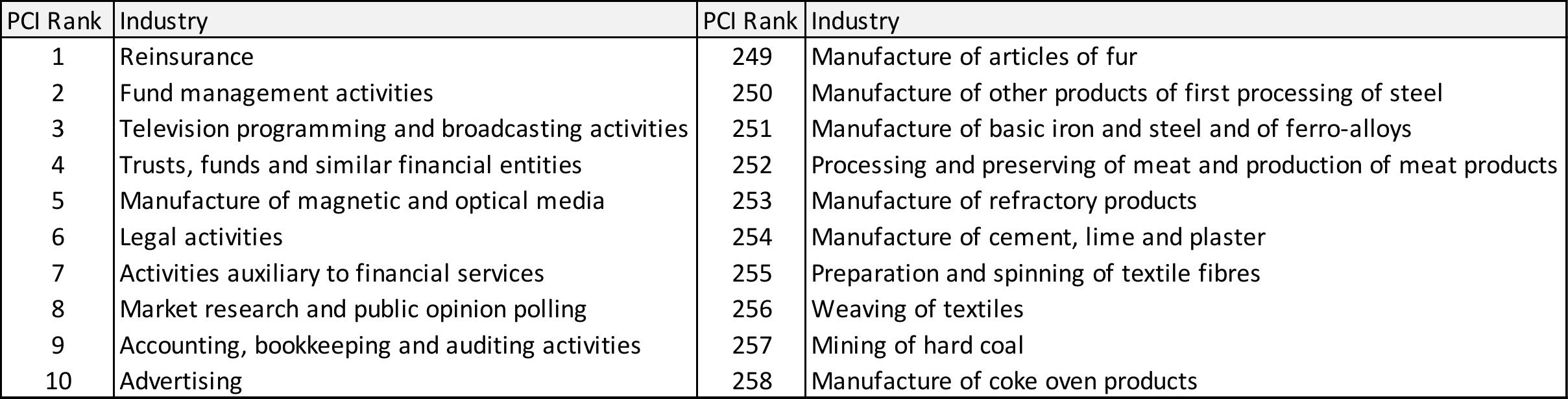}
	\caption{Top and bottom 10 industries ranked by PCI.}
	\label{tab:pci_industry_table}
\end{table}

\begin{table}[H]
	\centering
	\includegraphics[width=0.7\textwidth]{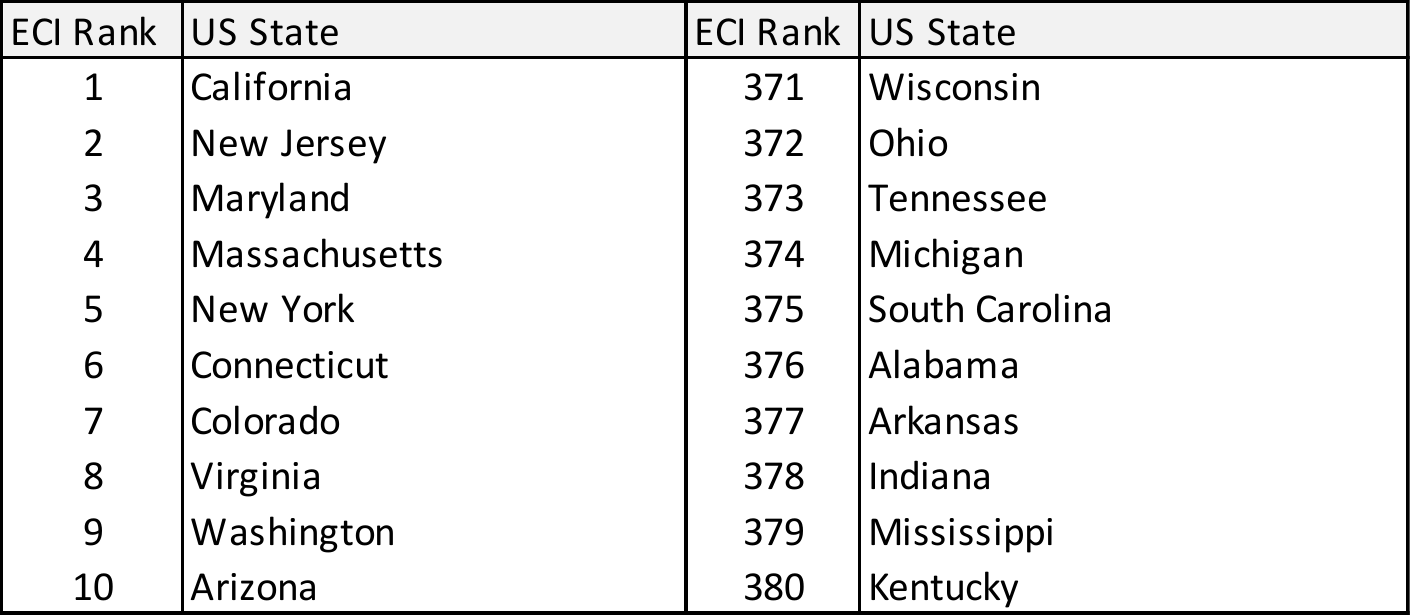}
	\caption{Top and bottom 10 US states ranked by ECI.}
	\label{tab:eci_state_table}
\end{table}

\begin{table}[H]
	\centering
	\includegraphics[width=1\textwidth]{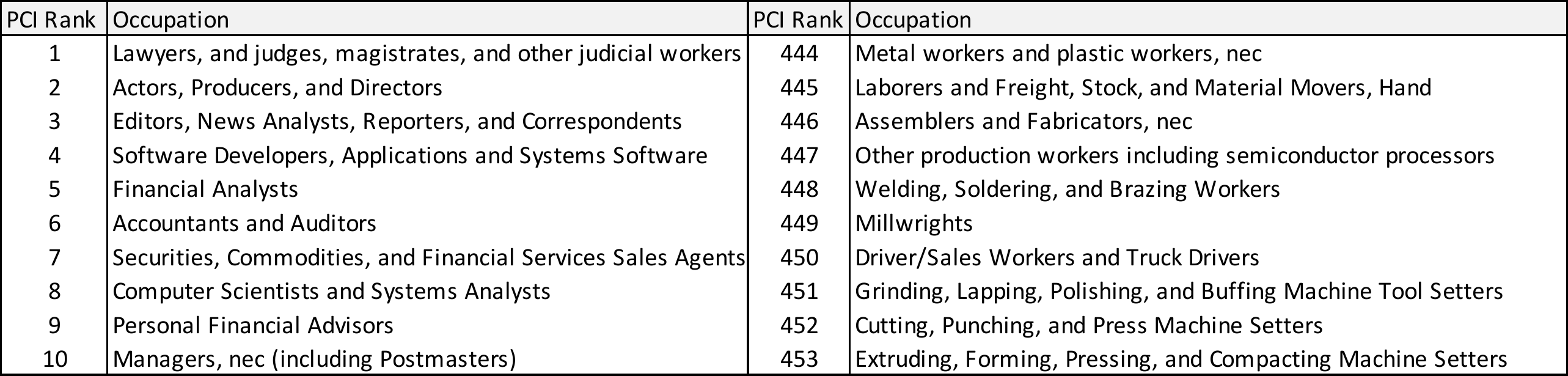}
	\caption{Top and bottom 10 occupations ranked by PCI.}
	\label{tab:pci_occ_table}
\end{table}

\section{Eigengap Heuristic Analysis}
In section \ref{sec:dimred} of the main paper, we showed that similarity networks constructed from the export and regional datasets did not partition well into two clusters. Here we analyse what is known as the \emph{eigengap heuristic}, which is a standard methodology used in spectral clustering analysis for determining the number of clusters present in the graph \cite{vonluxburg2007tutorial}.

The eigengap heuristic involves choosing the number of $k$ clusters such that the largest eigenvalues $\lambda_1, ..., \lambda_k$ of $\widetilde{M}$ are large, while $\lambda_{k+1}$ is relatively small. In Figure \ref{fig:eigengap_heuristic}, we show the largest six eigenvalues of the $\widetilde{M}$ matrix calculated for data on exports, UK regional industrial concentrations, and US state occupational concentrations respectively. In all three cases, the largest gap occurs between the first and second eigenvalue ($|\lambda_2 - \lambda_1|$). According to the eigengap heuristic, this suggests that from a spectral clustering perspective the graphs considered in this paper are likely to only contain one cluster. However, it is also important to note that the eigengap heuristic usually only works well if the data contains well-pronounced clusters - which is not the case here.

\begin{figure}[H]
	\centering
	\includegraphics[width=0.7\textwidth]{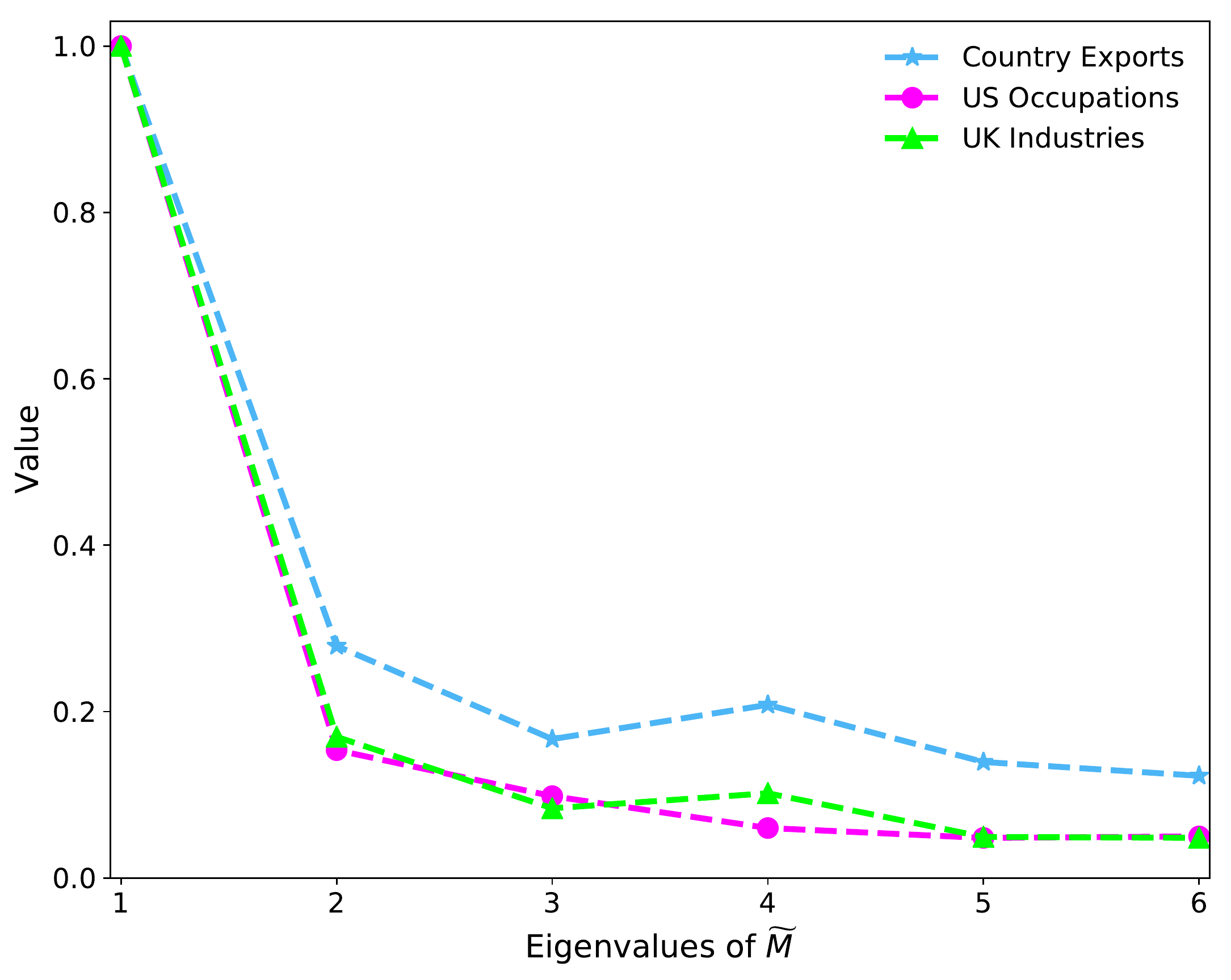}
	\caption{Top largest eigenvalues of the $\widetilde{M}$ matrix for data on exports, UK regional industrial concentrations and US state occupational concentrations.}
	\label{fig:eigengap_heuristic}
\end{figure}

\section{Robustness of empirical results to alternative RCA thresholds}
In principle, the use of the RCA measure to calculate the binary $M$ matrix can be particularly sensitive to the chosen threshold above which a country is considered to have a revealed comparative advantage in a product. For the empirical results shown in the main paper we have followed the most common approach and used a threshold of $1$. While the choice of threshold will have no bearing on the mathematical interpretation of the ECI and PCI, in this section we show to the extent to which empirical results for the country-export data are influenced by different RCA thresholds. 

In Panel A of Figure \ref{fig:panel_RCA_thresh_correlation}, we show how the empirical correlation between the ECI and per capita GDP change for different RCA thresholds. Correlations are highest between thresholds of 0.5 and 2. Panel B shows the correlation between the ECI and country diversity for different RCA thresholds. 

When the RCA threshold is zero, there are some products that are competitively exported by all countries. This means that the multiplicity of the largest eigenvalue is greater than one. In this case, since the eigenvector corresponding to the largest eigenvalue is proportional to diversity, the eigenvector corresponding to the second-largest eigenvalue is also proportional to diversity. Therefore, when the RCA threshold is zero, there is a perfect correlation between ECI and diversity.

\begin{figure}[H]
	\centering
	\includegraphics[width=1\textwidth]{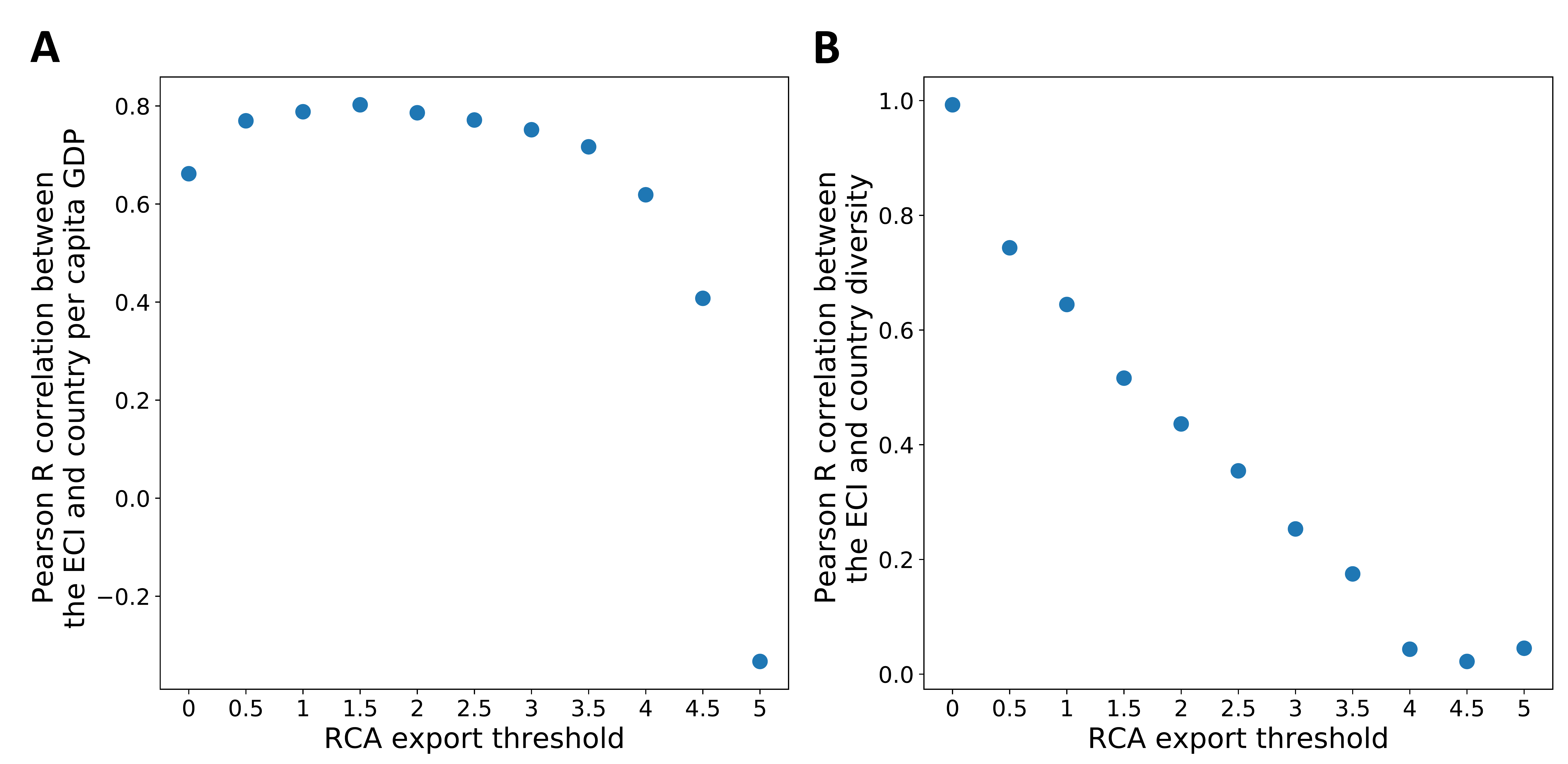}
	\caption{\textbf{Panel A)} Pearson correlation between the ECI and country per capita GDP for different RCA export thresholds. \textbf{Panel B)} Pearson correlation between the ECI and diversity for different RCA export thresholds}
	\label{fig:panel_RCA_thresh_correlation}
\end{figure}

Figure \ref{fig:panel_RCA_thresh_banding} examines how the pattern of specialization revealed by the ECI and PCI changes for different RCA thresholds. Here we compare binary $M$ matrices, each sorted by ECI and PCI, using RCA thresholds of $0.5$, $1$ and $2$. The pattern becomes more triangular for the lower RCA threshold (Panel A), largely because the ECI ordering is becoming closer to the ordering given by diversity (see Panel B of Figure \ref{fig:panel_RCA_thresh_correlation}). The higher RCA threshold (Panel C) shows a similar pattern of specialization to the original pattern shown in Panel B. 

\begin{figure}[H]
	\centering
	\includegraphics[width=1\textwidth]{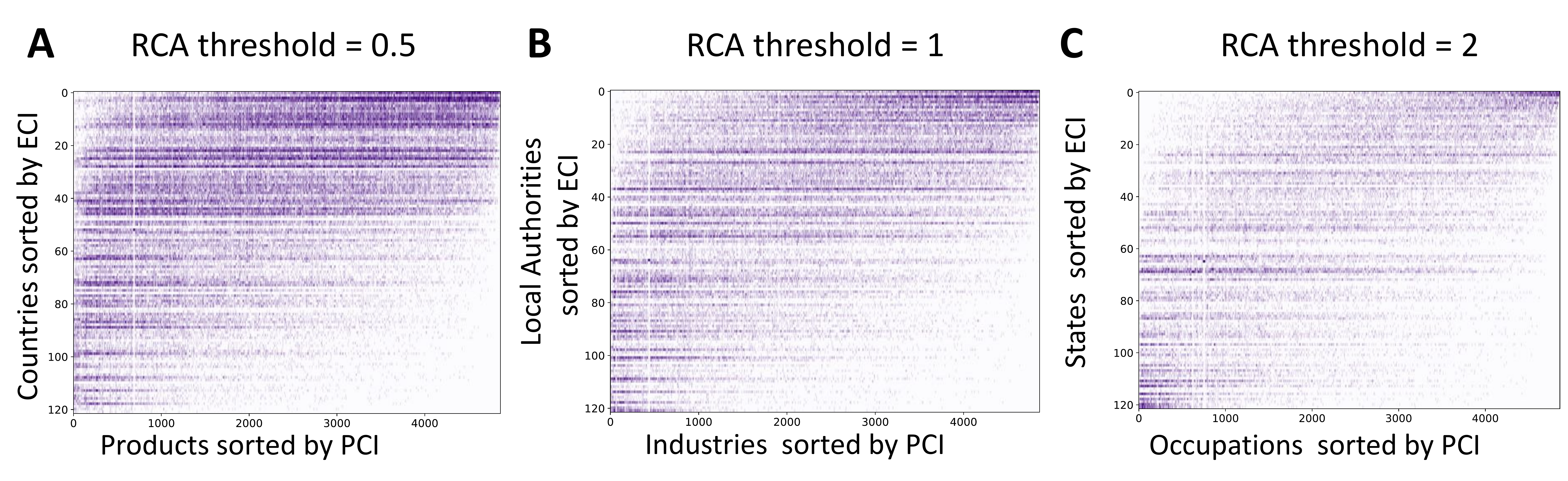}
	\caption{Country-product $M$ matrix with rows sorted by the ECI and columns sorted by the PCI constructed using different RCA thresholds:  \textbf{Panel A)} RCA threshold $=0.5$;  \textbf{Panel B)} RCA threshold $=1$;  \textbf{Panel C)} RCA threshold $=2$;}
	\label{fig:panel_RCA_thresh_banding}
\end{figure}

We also follow the approach taken in \cite{busto2012snestedness} and examine how the correlation between the ECI, per capita GDP and diversity change using a ``per capita'' version of RCA ($RCA\_POP$), given by

\begin{equation}\label{eq:RCA_pop}
RCA\_POP_{cp} = \frac{ x_{cp}/ n_{c}} { \sum_{c}x_{cp}/ \sum_{c} n_{c} },
\end{equation}
where $x_{cp}$ is country $c$'s exports of product $p$, $n_c$ is the population of country c and $M_{cp} = 0$ otherwise.

\begin{figure}[H]
	\centering
	\includegraphics[width=1\textwidth]{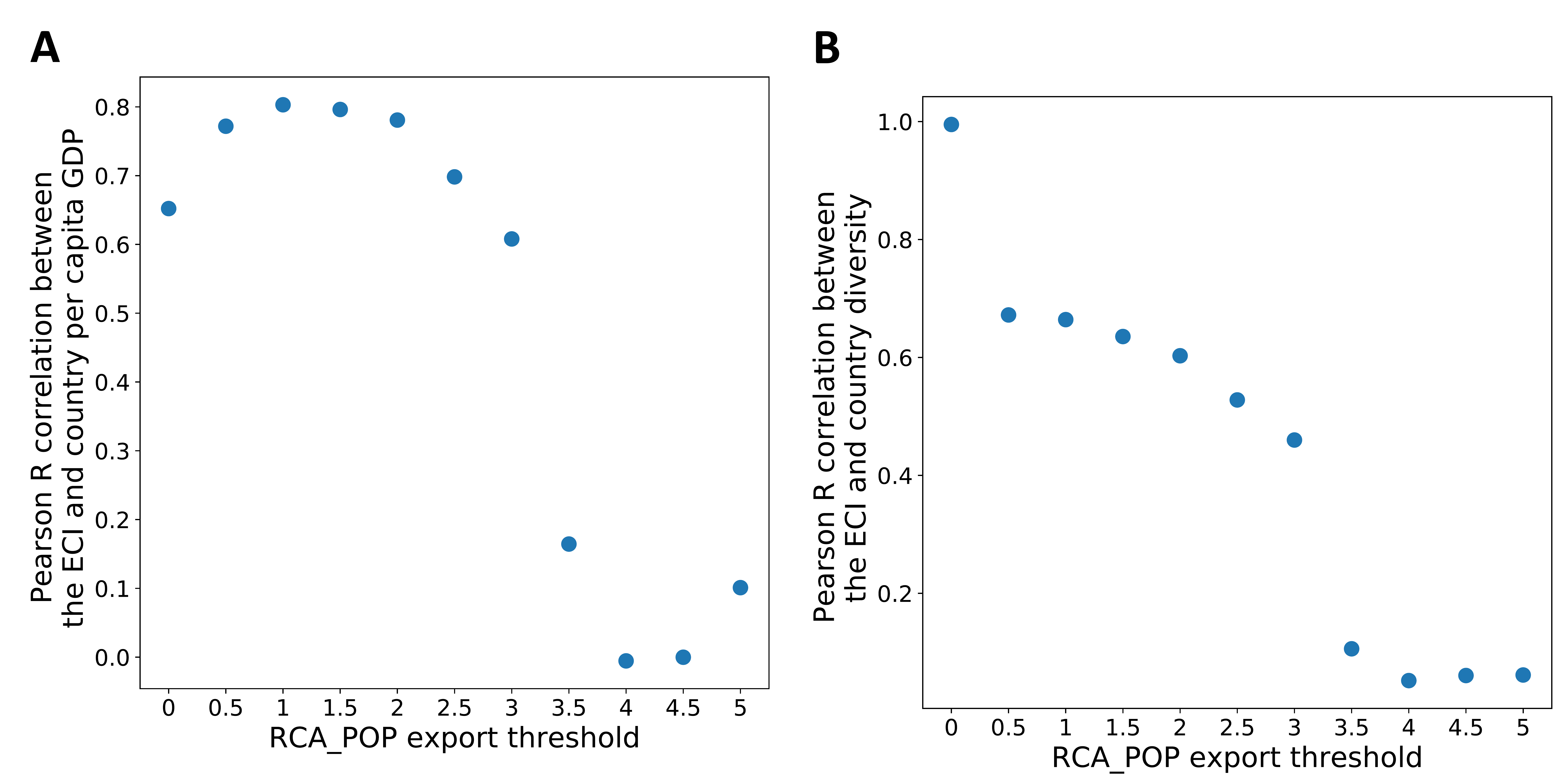}
	\caption{\textbf{Panel A)} Pearson correlation between the ECI and country per capita GDP for different per capita RCA thresholds. \textbf{Panel B)} Pearson correlation between the ECI and diversity for different per capita RCA  thresholds}
	\label{fig:panel_RCA_pop_thresh_correlation}
\end{figure}

Our results suggest that regardless of whether the per capita or original RCA version is applied, a threshold of 1 gives a strong correlation to per capita GDP. Moreover, the correlation between diversity and the ECI decreases as the RCA threshold is increased.

\newpage
\section*{Acknowledgements }
We would  like to thank Simon Angus, R. Maria Del Rio Chanona, Neave O'Clery, Devavrat Shah,  Muhammad Yildrim,  Ricardo Hausmann, Eric Kemp-Benedict, Luciano Pietronero and four anonymous reviewers for valuable comments and feedback on earlier drafts of this paper.

\subsection*{Funding}
This project was supported by Partners for the New Economy and the Oxford Martin School project on the Post-Carbon Transition.

\subsection*{Author contributions statement}

P.M. carried out the data analysis. All authors devised the research, wrote and revised the manuscript text. 
\subsection*{Data and materials availability}
All data needed to evaluate the conclusions in the paper is available from authors upon request.

\subsection*{Additional information}

\textbf{Competing financial interests} The authors declare that they have no competing interests

\end{document}